\documentclass{llncs}

\usepackage{amsmath,amsfonts,amssymb}
\usepackage{thm-restate}
\usepackage{semantic}
\usepackage{tikz}
\usetikzlibrary{arrows,automata}
\usepackage{wrapfig}
\usepackage{chngcntr}
\usepackage{microtype}
\usepackage{conclusive}
\newif\ifarxiv
\arxivtrue

\begin{document}
\title{Control strategies\\ for off-line testing of timed systems}
\author{L\'eo Henry \and Thierry J\'eron \and Nicolas Markey}
\institute{Univ. Rennes, INRIA \& CNRS, Rennes (France)}

\maketitle

\begin{abstract}
Partial observability and controllability are two well-known issues in
test-case synthesis for interactive systems.  We~address the problem
of partial control in the synthesis of test cases from timed-automata
specifications.  Building on the $\tioco$ timed testing framework,
we~extend a previous game interpretation of the test-synthesis problem
from the untimed to the timed setting.  This~extension requires a deep
reworking of the models, game interpretation and test-synthesis
algorithms.  We~exhibit strategies of a game that tries to minimize
both control losses and distance to the satisfaction of a test
purpose, and prove they are winning under some fairness assumptions.
This~entails that when turning those strategies into test cases,
we~get properties such as soundness and exhaustiveness of the test
synthesis method.

\end{abstract}

\section{Introduction} 
\label{sec:introduction}

Real-time interactive systems are systems interacting with their
environment and subject to timing constraints.  Such systems are
encountered in many contexts, in particular in critical applications
such as transportation, control of manufacturing systems,~etc.  Their
correctness is then of prime importance, but it is also very
challenging due to multiple factors: combination of discrete and
continuous behaviours, concurrency aspects in distributed systems,
limited observability of behaviours, or partial controllability of
systems.

One of the most-used validation techniques in this context is testing,
with variations depending on the design phases.  Conformance testing
is one of those variations, consisting in checking whether a real
system correctly implements its specification.  Those real systems are
considered as black boxes, thereby offering only partial
observability, for various reasons (\eg~because sensors cannot observe
all actions, or because the system is composed of communicating
components whose communications cannot be all observed, or again
because of intellectual property of peer software).  Controllability is
another issue when the system makes its own choices upon which the
environment, and thus
the tester, have a limited control.
One of the most-challenging activities in this context is the design
of 
test cases that, when executed on the real system, should
produce meaningful information about the conformance of the system at
hand with respect to its specification.  Formal models and methods are
a good candidate to help this test-case synthesis~\cite{Tre96}.

\looseness=-1
Timed Automata~(TA)~\cite{AD94} form a  class of model for the
specification of timed reactive systems.  It consists of automata
equipped with real-valued clocks where transitions between locations
carry actions, are guarded by constraints on clock values, and can
reset clocks.  TAs~are also equipped with invariants that constrain
the sojourn time in locations.  TAs~are popular in particular because
reachability of a location is decidable using symbolic representations
of sets of configurations by zones.  In the context of testing, it~is
adequate to refine TAs by explicitly distinguishing (controllable)
inputs and (uncontrollable) outputs, giving rise to TAIOs 
(Timed Automata with Inputs and Outputs). In the following, this model will be used for
most testing artifacts, namely specifications, implementations, and
test cases.  Since completeness of testing is hopeless in practice,
it~is helpful to rely on test purposes that describe those
behaviours that need to be tested because they are subject to errors.
In~our formal testing framework, an extension of TAIOs called Open TAIOs
(or~OTAIOs) is used to formally specify those behaviors.  OTAIOs~play
the role of observers of actions and clocks of the specification: they
synchronize on actions and clock resets of the specification (called
observed clocks), and control their proper clocks.  The~formal testing
framework also requires to formally define conformance as a relation
between specifications and their possible implementations.
In~the timed setting, the~classical~$\tioco$ relation~\cite{KT09} 
states that, after a timed observable trace of the
specification, the outputs and delays of the implementation should be
specified.

Test-case synthesis from TAs has been extensively studied during the
last 20 years
(see~\cite{COG98,CKL98,SVD01,EDK02,NS03,BB04,LMN04,KT09}, to cite
a~few).  As~already mentioned, one of the difficulties is partial
observation.  In~off-line testing, where the test cases are first
computed, stored, and later executed on the implementation, the tester
should anticipate all specified outputs after a trace.  In~the untimed
framework, this is tackled by determinization of the specification.
Unfortunately, this is not feasible for TAIO specifications since
determinization is not possible in general~\cite{AD94,Fin06}.  The
solution was then either to perform on-line testing where subset
construction is made on the current execution trace, or to restrict to
determinizable sub-classes.  More recently, some advances were
obtained in this context~\cite{BJSK12} by the use of an approximate
determinization using a game approach~\cite{BSJK15} that preserves
$\tioco$ conformance.  Partial observation is also dealt with
by~\cite{DLLMN10} with a variant of the TA model where observations
are described by observation predicates composed of a set of locations
together with clock constraints.  Test cases are then synthesized as
winning strategies, if~they exist, of~a~game between the specification
and its environment that tries to guide the system to satisfy the test
purpose.

\looseness=-1
The problem of test synthesis is often informally presented as a game
between the environment and the system (see \eg~\cite{Yan04}).
But very few papers effectively
take into account the controllability of the system.
In~the~context of testing for timed-automata models~\cite{DAVID-DATE-2008}, proposes a game approach where test cases are winning strategies of a reachability game.
But this is restricted to deterministic models and controllability is not really taken into account.
In~fact, like in~\cite{DLLMN10}, 
 the game is abandonned when control is lost, and it is suggested to modify the test purpose in this case.
This~is mitigated in~\cite{DAVID-MBT-2008} with cooperative strategies,
which rely on the cooperation of the system under test to win the game.
A~more convincing approach to the control problem is the one
of~\cite{Ram98} in the untimed setting, unfortunately a quite
little-known work.
The~game problem consists in satisfying the test purpose (a~simple
sub-sequence), while trying to avoid control losses occurring when
outputs offered by the system leave this behaviour.  The~computed
strategy is based on a rank that measures both the distance to the
goal and the controls losses.

The current paper adapts the approach proposed in~\cite{Ram98}
to the timed context using the framework developed in~\cite{BJSK12}.
Compared to~\cite{Ram98}, the model of TA is much more complex than transition systems, the test purposes are
also much more powerful than simple sub-sequences, thus even if the
approach is similar, the game has to be completely revised.  Compared
to~\cite{DLLMN10}, our model is a bit different since we do not rely on
observation predicates, but partial observation comes from internal
actions and choices.  We~do not completely tackle non-determinism since
we assume determinizable models at some point.
In~comparison, \cite{DLLMN10}~avoids determinizing~TAs,
relying on the determinization of a finite state model,
thanks to a projection on a finite set of observable predicates.
Cooperative strategies of~\cite{DAVID-MBT-2008} have similarities with our fairness assumptions, but their models are assumed deterministic.
Our approach takes controllability into account
in a more complete and practical way 
with the reachability game and rank-lowering strategies.

The paper is organized as follows. Chapter~2 introduces basic models:
TAs, TAIOs and their open counterparts OTAs, OTAIOs, and then timed
game automata (TGA).  Chapter~3 is dedicated to the testing framework
with hypothesis on models of testing artifacts, the conformance
relation and the construction of the \emph{objective-centered tester}
that denotes both non-conformant traces and the goal to reach
according to a test purpose.  Chapter~4 constitutes the core of the
paper. The test synthesis problem is interpreted as a game on the
objective-centered tester.  Rank-lowering strategies are proposed as
candidate test cases, and a fairness assumption is introduced to make
such strategies win.  Finally properties of test cases with respect to
conformance are proved.

\ifarxiv\else
By lack of space, not all proofs could be included. They~can be found in~\cite{arxiv}.
\fi


\section{Timed automata and timed games} 
\label{sec:models}
In this section, we introduce our models for timed systems and
for concurrent games on these objects, along with some useful notions
and operations.

\subsection{Timed automata with inputs and outputs} 
\label{sub:timed_automata}
Timed automata (TAs)~\cite{AD94} are one of the most widely-used classes
of models for reasoning about computer systems subject to real-time
constraints. Timed automata are finite-state automata augmented with
real-valued variables (called~\emph{clocks}) to constrain the
occurrence of transitions along executions.
In~order to adapt these models to the testing framework,
we~consider TAs with inputs and outputs (TAIOs), in which the
alphabet is split between input, output and internal actions
(the~latter being used to model partial observation). We~present the
\emph{open} TAs (and open TAIOs)~\cite{BJSK12}, which
allow the models to observe and synchronize with
a set of non-controlled clocks.

Given a finite set of \emph{clocks}~$X$, a~\emph{clock valuation}
over~$X$ is a function $v\colon X \to \bbR_{\geq 0}$.
We~note $\overline{0}_X$ (and often omit to mention~$X$ when clear
from the context) for the valuation assigning $0$ to all clocks
in~$X$.  Let~$v$ be a clock valuation,
for any~$t\in \bbR_{\geq 0}$, we~denote with $v + t$ the valuation mapping each clock $x\in X$
to~$v(x)+t$, 
and for a subset~$X'\subseteq X$, we~write
$v_{[X'\leftarrow 0]}$ for the valuation mapping all clocks in~$X'$
to~$0$, and all clocks in~$X\setminus X'$ to their values in~$v$.

A \emph{clock constraint} is a finite conjunction of atomic
constraints of the form $x\sim n$ where $x\in X$, $n\in\bbN$, and
$\mathord\sim \in\{\mathord{<},\mathord{\leq}, \mathord{=},
\mathord{\geq},\mathord{>}\}$. That a valuation~$v$ satisfies a clock
constraint~$g$, written~$v\models g$, is defined in the obvious way.
We~write $\ClC(X)$ for the set of clock constraints over~$X$.

\begin{defn}{}{OTAIO}
  An \emph{open timed automaton}~(OTA) is a tuple\footnote{For this and the following definitions, we~may omit to mention superscripts when the corresponding automaton is clear from the context.}
  $\mA=(L^\mA,l_0^\mA,\Sigma^\mA,\penalty1000
  X_p^\mA\uplus X_o^\mA,\penalty1000
  I^\mA,E^\mA)$ where:
\begin{itemize}
  \item $L^\mA$ is a finite set of \emph{locations}, with $l_0^\mA\in
    L^\mA$ the \emph{initial location},
  \item $\Sigma^\mA$ is a finite alphabet,
  \item $X^\mA=X_p^\mA\uplus X_o^\mA$ is a finite set of clocks,
    partitionned into \emph{proper clocks}~$X_p^\mA$ and
    \emph{observed clocks}~$X_o^\mA$; only proper clocks may be reset
    along transitions. 
  \item $ I^\mA\colon L^\mA \rightarrow \ClC(X^\mA)$ assigns invariant
    constraints to locations.
  \item $E^\mA \subseteq L^\mA \times \ClC(X^\mA) \times \Sigma^\mA \times
    2^{X_p^\mA} \times L^\mA$ is a finite set of
    \emph{transitions}.
    For $e=(l,g,a,X',l')\in E^\mA$, we~write $\act(e)=a$.
\end{itemize}
An Open Timed Automaton with Inputs and Outputs (OTAIO) is an OTA in which
$\Sigma^\mA=\Sigma_?^\mA\uplus \Sigma_!^\mA \uplus \Sigma_\tau^\mA$ is
the disjoint union of \emph{input actions} in~$\Sigma_?^\mA$ (noted
$?a$, $?b$,~...), \emph{output actions} in~$\Sigma_!^\mA$ (noted $!a$,
$!b$,~...), and \emph{internal actions} in~$\Sigma_\tau^{\mA}$ (noted
$\tau_1$, $\tau_2$, ...)  We~write $\Sigma_{\obs} =
\Sigma_?\uplus\Sigma_!$ for the alphabet of \emph{observable} actions.
Finally, a Timed Automaton~(TA) (resp. a Timed Automaton with Inputs
and Outputs~(TAIO)) is an~OTA (resp. an~OTAIO) with no observed
clocks.
\end{defn}
TAIOs will be sufficient to model most objects, but the ability of
OTAIOs to observe other clocks will be essential for test
purposes (see~Section~\ref{sub:test_context}), which need to
synchronize with the specification.

Let $ \mA = (L, l_0, \Sigma,
X_p\uplus X_o,I,E)$ be an OTA. Its~\emph{semantics} is
defined as an infinite-state transition system $\mT^{\mA} = (S^\mA, s_0^\mA, \Gamma^\mA,
\rightarrow^\mA)$ where:
\begin{itemize}
\item $S^\mA = \{(l,v)\in L \times \bbR^{X}_{\geq 0 }\mid v\models I(l)\}$ is the
  (infinite) set of \emph{configurations},
  with initial configuration $s_0^\mA = (l_0, \overline{0}_X)$.
\item $\Gamma^\mA = \bbR_{\geq 0} \uplus (E\times 2^{X_o})$ is
  the set of \emph{transitions labels}.
\item $\mathord{\rightarrow^\mA} \subseteq S^\mA \times \Gamma^\mA \times S^\mA$
is the
  \emph{transition relation}. It~is defined as the union of
  \begin{itemize}
  \item the set of transitions corresponding to \emph{time elapses}:
    it contains all triples $((l,v),\delta,(l',v'))\in S^\mA\times
    \bbR_{\geq 0}\times S^\mA$ for which $l=l'$ and
    $v'=v+\delta$. By~definition of~$S^\mA$, both $v$ and~$v'$ satisfy the
    invariant~$I(l)$.
    
  \item the set of transitions corresponding to \emph{discrete moves}:
    it~contains all triples $((l,v),(e,X'_o),(l',v'))\in S^\mA\times
    (E\times 2^{X_o})\times S^\mA$ such that, writing
      $e=(m,g,a,X'_p,m')$, it~holds $m=l$, $m'=l'$, $v\models g$, and
      $v'=v_{[X'_p\cup X'_o\leftarrow 0]}$. Again, by~definition, $v\models I(l)$
      and $v'\models I(l')$.
  \end{itemize}
\end{itemize}

An OTA has no control over its observed clocks,
the intention being to synchronize them later in a product (see Def.~\ref{ProductOTAIO}).
Hence, when a discrete transition is taken, any~set~$X'_o$ of observed clocks may be
reset. 
When~dealing with plain TAs, where $X_o$ is~empty, we~may
write $(l,v)\xrightarrow{e}(l',v')$ in place of
$(l,v)\xrightarrow{(e,\emptyset)}(l',v')$.

A~\emph{partial run}
of $\mA$ is a (finite or infinite) sequence of
transitions in $\mT^{\mA}$
$\rho=((s_i,\gamma_i,s_{i+1}))_{1\leq i<n}$, with
$n\in\bbN\cup\{+\infty\}$.
We~write $\first(\rho)$ for~$s_1$ and, when $n\in\bbN$, $\last(\rho)$
for~$s_{n}$. A~\emph{run} is a partial run starting in the initial
configuration~$s_0^\mA$.  
The~duration of~$\rho$ is $\duration(\rho)=\sum_{\gamma_i\in\bbR_{\geq 0}} \gamma_i$.
In the sequel, we~only consider TAs in which any infinite run has
infinite duration.
We~note~$\Ex(\mA)$
for the set of runs of~$\mA$
and $\pEx(\mA)$ the subset of partial runs.

State~$s$ is \emph{reachable} from state~$s'$
when there exists a
partial run from~$s'$ to~$s$. We~write $\Reach(\mA,S')$
for the set of states that are reachable
from some state in~$S'$, and $\Reach(\mA)$ for
$\Reach(\mA, \{s_0^\mA\})$.

The~\emph{(partial) sequence} associated with a (partial) run~$\rho=((s_i,\gamma_i,s'_i))_i$ is
$\seq(\rho)=(\proj(\gamma_i))_i$, where
$\proj(\gamma)=\gamma$ if $\gamma\in\bbR_{\geq 0}$, and 
$\proj(\gamma)=(a,X'_p\cup X'_o)$ if $\gamma=((l,g,a,X'_p,l'),X'_o)$.
We~write  $\pSeq(\mA)= \proj(\pEx(\mA))$ and $\Seq(\mA)=\proj(\Ex(\mA))$ for the sets of (partial) sequences of~$\mA$.
We~write $s\xrightarrow{\mu}s'$ when there exists a (partial) finite
run~$\rho$ such that $\mu=\proj(\rho)$,
$\first(\rho)=s$ and $\last(\rho)=s'$, and~write $\duration(\mu)$ for $\duration(\rho)$.
We~write $s\xrightarrow{\mu}$ when $s\xrightarrow{\mu} s'$ for some~$s$.

If $\mA$ is a TAIO, 
the \emph{trace} of a (partial) sequence corresponds to what can be
observed by the environment, namely delays and observable actions. The
trace of a sequence is the limit of the following inductive definition, for
$\delta_i\in\mathbb{R}_{\geq0}$, $a\in\Sigma_{obs}$,
$\tau\in\Sigma_\tau$, $X' \subseteq X$, and a~partial sequence~$\mu$:
\begin{xalignat*}1
  \Trace(\delta_1...\delta_k) &=   \textstyle\sum^k_{i=1}\delta_i
\qquad\text{ (in particular $\Trace(\epsilon)=0$)}
  \\
  \Trace(\delta_1...\delta_k.(\tau,X').\mu)&= \textstyle(\sum^k_{i=1}\delta_i)\cdot \Trace(\mu)
  \\
  \Trace(\delta_1...\delta_k.(a,X').\mu)&= \textstyle(\sum^k_{i=1}\delta_i)\cdot a \cdot \Trace(\mu)
\end{xalignat*}
We note $\Traces(\mA)=\Trace(\Seq(\mA))$ the set of traces
corresponding to runs of~$\mA$ and $\pTraces(\mA)$ the subset of traces
corresponding to partial runs.
Two~OTAIOs are said to be \emph{trace-equivalent} 
if they have the same sets of traces.  We~furthermore define, for an
OTAIO~$\mA$, a~trace~$\sigma$ and a configuration~$s$:
\begin{itemize}
\item $\mA\after \sigma = \{ s\in S \mid \exists \mu\in \Seq(\mA),
  s_0\xrightarrow{\mu}s\wedge \Trace(\mu)=\sigma\}$ is the set of all
  configurations that can be reached when~$\sigma$ has been
  observed from~$s_0^\mA$.
\item $\enab(s)=\{ e\in E^\mA\mid s\xrightarrow{e}\}$ is the set of
  transitions enabled in~$s$.
\item $\elapse(s)=\{t\in\mathbb{R}_{\geq 0} \mid \exists
  \mu\in(\mathbb{R}_{\geq 0}\cup(\Sigma_\tau\times 2^X))^{*},
  s\xrightarrow{\mu} \wedge  \duration(\mu) = t\}$ is the set of delays
  that can be observed from location $s$ without any observation.
\item $\Sout(s) = \{a\in \Sigma_!\mid \exists e\in \enab(s), \act(e)=a\}\cup\elapse(s)$
  is the set of possible
  outputs and delays that can be observed from~$s$. For~$S' \subseteq
  S$, we note $\Sout(S')= \bigcup_{s\in S'} \Sout(s)$.
\item $\Sin(s)=\{a\in \Sigma_!\mid \exists e\in \enab(s), \act(e)=a\}$
  is the set of possible inputs that can be
  proposed when arriving in~$s$. For~$S' \subseteq S$, we~note
  $\Sin(S')= \bigcup_{s\in S'} \Sin(s)$.
\end{itemize}

We now define some useful sub-classes of OTAIOs.
  An OTAIO~$\mA$ is said
  \begin{itemize}
  \item \emph{deterministic} if for all $\sigma\in Traces(\mA)$,
    $\mA\after \sigma$ is a singleton;
    \item \emph{determinizable} if there exists a trace-equivalence
      deterministic OTAIO;
    \item \emph{complete} if $S=L\times\bbR_{\geq 0}^X$
      (i.e., all invariants are always true)
      and for any~$s\in S$ and
      any~$a\in\Sigma$, it~holds $s\xrightarrow{a,X'}$ for
      some~$X'\subseteq X$;
    \item \emph{input-complete} if for any $s\in\Reach(\mA)$,
    $\Sin(s)=\Sigma_?$;
    \item \emph{non-blocking} if for any $s\in\Reach(\mA)$ and any
      non-negative real~$t$, there is a partial run~$\rho$ from~$s$
      involving no input actions (i.e., $\proj(\rho)$ is a sequence
      over $\bbR_{\geq 0} \cup (\Sigma_!\cup\Sigma_\tau)\times 2^X$)
      and such that $\duration(\rho)=t$;
    \item \emph{repeatedly observable} if for any $s\in\Reach(\mA)$,
      there exists a partial run~$\rho$ from~$s$
      such that $\Trace(\rho)\notin\bbR_{\geq 0}$.
  \end{itemize}

The product of two OTAIOs extends the classical product of~TAs.
\begin{defn}{}{ProductOTAIO}\label{ProductOTAIO}
  Given two OTAIOs $\mA = (L^\mA, l_0^\mA,
  \Sigma_?\uplus\Sigma_!\uplus\Sigma_\tau,X_p^\mA\uplus
  X_o^\mA,I^\mA,E^\mA)$ and $\mB = (L^\mB, l_0^\mB,
  \Sigma_?\uplus\Sigma_!\uplus\Sigma_\tau,X_p^\mB\uplus X_o^\mB,I^\mB,E^\mB)$
  over the same alphabets, their \emph{product} is the OTAIO
  $\mA\times\mB = (L^\mA\times L^\mB, (l_0^\mA,l_0^\mB),
  \Sigma_?\uplus\Sigma_!\uplus\Sigma_\tau,(X_p^\mA\cup X_p^\mB)\uplus((X^\mA_o\cup
  X^\mB_o)\setminus(X_p^\mA\cup X_p^\mB)),I,E)$ where
  $I\colon (l_1,l_2)\mapsto
  I^\mA(l_1)\wedge I^\mB(l_2)$ and $E$
  is the (smallest) set such that for each
  $(l^1,g^1,a,X_p'^1,l'^1)\in E^\mA$ and $(l^2,g^2,a,X_p'^2,l'^2)\in
  E^\mB$, $E$~contains $((l^1,l^2),g^1\wedge g^2, a, X_p'^1\cup
  X_p'^2, (l'^1,l'^2))$.
\end{defn}
The product of two OTAIOs corresponds to the intersection of the sequences of the orginal OTAIOs, \ie~$\Seq(\mA\times\mB)=\Seq(\mA)\cap\Seq(\mB)$~\cite{BSJK15}.


\subsection{Timed games} 
\label{sub:timed_games}

We introduce timed game automata~\cite{AMPS98},
which we later use to turn the test artifacts into games between 
the tester (controlling the environment) and the implementation, on an arena constructed from the specification.
\begin{defn}{}{TGA}
A \emph{timed game automaton} (TGA) is a timed automaton $ \g = (L,
l_0, \Sigma_c\uplus\Sigma_u,X,I,E)$ where $\Sigma=\Sigma_c\uplus
\Sigma_u$ is partitioned into actions that are
controllable~($\Sigma_c$) and uncontrollable~($\Sigma_u$) by the
player.
\end{defn}
All the notions of runs and sequences 
defined previously for TAs are extended to TGAs,
with the interpretation of $\Sigma_c$ as inputs and $\Sigma_u$ as outputs.  
\begin{defn}{}{Strategy}
Let $ \g = (L, l_0, \Sigma_c\uplus\Sigma_u,X, I,E)$ be a
TGA. A~\emph{strategy} for the player is a partial function $f\colon
Ex(\g) \to \bbR_{\geq0}\times(\Sigma_c\cup\{\bot\})\setminus\{(0,\bot)\}$
such that for any finite run~$\rho$, letting $f(\rho)=(\delta,a)$,
$\delta \in \elapse(\last(\rho))$ is a possible delay from $\last(\rho)$,
and there is an
$a$-transition available from the resulting configuration (unless
$a=\bot$).
\end{defn}

Strategies give rise to sets of runs of~$\g$, defined as follows:
\begin{defn}{}{Outcome}
  Let $ \g = (L, l_0, \Sigma,X,I,E)$ be a TGA, $f$ be a strategy over
  $\g$, and $s$ be a configuration.  The~\emph{set of outcomes} of~$f$
  from~$s$, noted $\Outcome(s,f)$, is the smallest subset of partial
  runs starting from~$s$ containing the empty partial run from~$s$
  (whose~last configuration is~$s$),
  and s.t. for any~$\rho\in\Outcome(s,f)$,
  letting $f(\rho)=(\delta,a)$ and $\last(\rho)=(l,v)$, we~have
  \begin{itemize}
  \item $\rho\cdot ((l,v),\delta,(l,v+\delta')) \cdot
    ((l,v+\delta'),e,(l',v')) \in \Outcome(s,f)$ for any $0\leq
    \delta'\leq \delta$ and $\act(e)\in\Sigma_u$ such that
    $((l,v+\delta'),e,(l',v')) \in \pEx(\mA)$;
  \item and
    \begin{itemize}
    \item either $a=\bot$, and $\rho\cdot ((l,v),\delta,(l,v+\delta)) \in
      \Outcome(s,f)$;
    \item or $a\in\Sigma_c$, and
       $\rho\cdot ((l,v),\delta,(l,v+\delta)) \cdot
      ((l,v+\delta),e,(l',v')) \in \Outcome(s,f)$ with $\act(e)=a$;
    \end{itemize}
  \end{itemize}
  An infinite partial run is in $\Outcome(s,f)$ if infinitely many of its finite
  prefixes~are.
\end{defn}

In this paper, we will be interested in reachability winning
conditions (under particular conditions).  In the classical setting,
the set of winning configurations can be computed iteratively,
starting from the target location and computing controllable
predecessors in a backward manner. The~computation can be performed on
regions, so that it terminates (in exponential
time)~\cite{AMPS98,CDF+05}. We~extend this approach to our
test-generation framework in
Section~\ref{sec:translating_objectives_into_games}.

\section{Testing framework} 
\label{sec:testing_framework}
We now present the testing framework, defining \textit{(i)}~the main testing
artifacts \ie~ specifications, implementations, test purposes,
and test cases, along with the assumptions on them; \textit{(ii)}  a conformance
relation relating implementations and specifications. The combination of the
test purposes and the specification and the construction of an
approximate deterministic tester is afterward explained.
\subsection{Test context} 
\label{sub:test_context}
We~use TAIOs as models for specifications, implementations and test
cases, and OTAIOs for test purposes.  This allows to define liberal
test purposes, and on a technical side, gives a unity to the manipulated
objects.

In order to enforce the occurrence of conclusive verdicts, we~equip
specifications with \emph{restart transitions}, corresponding to a system
shutdown and restart, and assume that from any (reachable) configuration, a
restart is always reachable.

\begin{defn}{}{Specification}
A \emph{specification with restarts} (or simply \emph{specification})
on $(\Sigma_?,\Sigma_!,\Sigma_\tau)$ is a non-blocking,
repeatedly-observable TAIO $\mS = (L^\mS, l_0^\mS,
(\Sigma_?\cup\{\zeta\})\uplus\Sigma_!\uplus\Sigma_\tau,
X_p^\mS,I^\mS,E^\mS)$ where $\zeta\notin\Sigma_?$ is the restart
action. We~let $\Restart^\mS= E^\mS \cap (L^\mS\times
G_{M^\mS}(X^\mS)\times\{\zeta\}\times \{X_p^\mS\}\times \{l_0^\mS\})$ be
the set of $\zeta$-transitions, and it is assumed that from any
reachable configuration, there exists a finite partial execution
containing~$\zeta$, i.e. for any $s \in \Reach(\mS)$, there exists
$\mu$ s.t. $s\xrightarrow{\mu\cdot \zeta} s_0^\mS$.
\end{defn}

The non-blocking hypothesis rules out "faulty" specifications having
no conformant physically-possible implementation.
Repeated-observability will be useful for technical
reasons, when analyzing the exhaustiveness property of test cases.
Our assumption on $\zeta$-transitions entails:
\begin{restatable}{proposition}{propSC}\label{pr:Asc_s}
Let $\mS$ be a specification with restarts. Then $Reach(\mT_\mS)$ is
strongly-connected.
\end{restatable}

\begin{wrapfigure}{R}{0.55\textwidth}
\centering
\begin{tikzpicture}[yscale=.8,->,>=stealth',shorten >=1pt,auto,node distance=3.3cm,semithick,  scale=0.75,  every node/.style={transform shape}	]
  \tikzstyle{every state}=[fill=red!50,draw=none,text=black, font=\scriptsize,inner sep=0pt, minimum size=1.1cm]
  \tikzstyle{every path}=[font=\scriptsize, below]
  \node[state] (End1)     [] {$\nodeinv{\Dest_1}{\true}$};
  \node[state] (Boarding) [right of=End1] {$\nodeinv{\Boarding}{x\leq 3}$};
  \node[state] (End2)     [right of=Boarding] {$\nodeinv{\Dest_2}{\true}$};
  \node[state] (Sort)     [below of=Boarding] {$\nodeinv{\Sort}{x\leq 1}$};
  \node[state] (Waste)    [right of=Sort, node distance=2cm] {$\nodeinv{\Waste}{x\leq 1}$};
  \node[state,initial] (Start)    [below of=Sort] {$\nodeinv{\Start}{x\leq 2}$};

\path (Start)     edge[bend right=0] node[left] {$\stack{x\leq2}{\tau}{\{x\}}$} (Sort)
      (Sort)      edge node[above] {$\stack{\true}{\waste!}{\{x\}}$} (Waste)
      (Sort)      edge node[right] {$\stack{\true}{\tau}{\{x\}}$}   (Boarding)
      (Boarding)  edge node[above] {$\stack{\true}{\ship_1?}{\{x\}}$} (End1)
      (Boarding)  edge node[above] {$\stack{\true}{\ship_2?}{\{x\}}$} (End2)
      (Boarding)  edge [bend right=40] node [pos=.3,left] {$\stack{x=3}{\past!}{\{x\}}$}   (Start)
      (End1)      edge [loop above] node [pos=.25,left] {$\stack{x=1}{\aend_1!}{\{x\}}$}  (End1)
      (End2)      edge [loop above] node [pos=.25,left] {$\stack{x=1}{\aend_2!}{\{x\}}$}  (End2)
      (End1)      edge [bend right] node[left] {$\stack{\true}{\zeta}{\{x\}}$}  (Start.170)
      (End2)      edge[bend left] node[right] {$\stack{\true}{\zeta}{\{x\}}$}  (Start.10)
      (Waste)     edge [loop above] node [pos=.25,above] {$\stack{x=1}{\tau}{\{x\}}$}   (Waste)
      (Waste)     edge node[right] {$\stack{\true}{\zeta}{\{x\}}$}  (Start);
\end{tikzpicture}
	\caption{A conveyor belt specification.}
	\label{fig:conv_belt}
\end{wrapfigure}
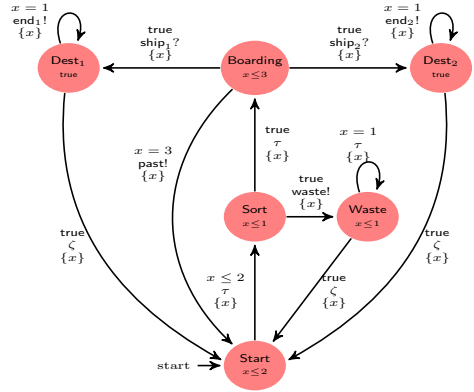
\begin{example}
\label{ex:exemple_of_specification}
Figure~\ref{fig:conv_belt} is an example of specification for a
conveyor belt. After a maximum time of 2 units (depending for example
on their weight),
packages reach a sorting point where they are
automatically sorted between packages to reject and packages to
ship. Packages to reject go to waste, while packages to ship are
sent to a boarding platform, where an operator can send them to two
different destinations. If~the~operator takes more than 3 units of
time to select a destination, the~package goes past the boarding
platform and restarts the process.
\end{example}

In practice, test purposes are used to describe the intention of test
cases, typically behaviours one wants because they must be correct
and\slash or an error is suspected.  In~our~formal testing framework,
we~describe them with OTAIOs that observe the specification together
with accepting locations.

\begin{defn}{}{TestPurpose}
  Given a specification $\mS = (L^\mS, l_0^\mS,
  (\Sigma_?\cup\{\zeta\})\uplus\Sigma_!\uplus\Sigma_\tau,X_p^\mS,I^\mS,E^\mS\uplus\Restart)$, a~\emph{test purpose for~$\mS$} is a pair
  $(\TP,\Accept^\TP)$ where
  $\TP = (L^\TP, l_0^\TP,
  \Sigma_?\cup\{\zeta\}\uplus\Sigma_!\uplus\Sigma_\tau,X_p^\TP\uplus
  X_p^\mS,I^\TP,E^\TP)$ is a complete OTAIO together with a subset
  $\Accept^\TP\subseteq L^\TP$ of accepting locations, and such that
  transitions carrying restart actions~$\zeta$ reset all proper clocks
  and return to the initial state (i.e., for any
  $\zeta$-transition~$(l,g,\zeta,X',l')\in E$, it~must be $X'=X_p^\TP$ and
  $l'=l_0^\TP$).
\end{defn}
In the following, we may simply write~$\TP$ in place of
$(\TP, \Accept^\TP)$.  We~force test purposes
to be complete because they should never constrain the runs of the
specification they observe, but should only label the accepted
behaviours to be tested.  Test purposes observe exactly the clocks
of the specification in order to synchronize with them, but cannot
reset them.

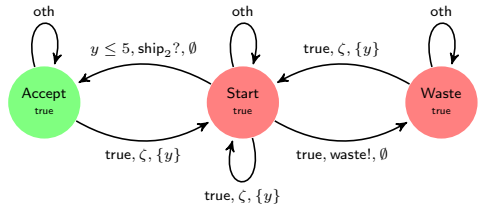
\begin{wrapfigure}{R}{0.57\linewidth}
\centering
\begin{tikzpicture}[->,>=stealth',shorten >=1pt,auto,node distance=3.3cm,semithick,  scale=0.8,  every node/.style={transform shape}]
  \tikzstyle{every state}=[fill=red!50,draw=none,text=black, font=\scriptsize, inner sep=0pt, minimum size=1.2cm]
  \tikzstyle{every path}=[font=\scriptsize, sloped, below]
  \node[state] (Start)     [] {$\nodeinv{\Start}{\true}$};
  \node[state,fill=green!50] (Accept) [left of=Start] {$\nodeinv{\Accept}{\true}$};
  \node[state] (Waste)     [right of=Start] {$\nodeinv{\Waste}{\true}$};
  \path[use as bounding box] (0,0) -- (0,-2);
\path (Start)     edge [bend right] node [above] {$y\leq5,\ship_2?,\emptyset$} (Accept)
      (Start)     edge [bend right] node {$\true,\waste!,\emptyset$} (Waste)
      (Start)     edge [loop above] node {$\oth$}   (Start)
      (Start)     edge [loop below] node[below] {$\true,\zeta,\{y\}$} (Start)
      (Accept)    edge [loop above] node {$\oth$} (Accept) 
      (Accept)    edge [bend right] node {$\true,\zeta,\{y\}$} (Start)
      (Waste)    edge [loop above] node {$\oth$} (Waste) 
      (Waste)    edge [bend right] node[above] {$\true,\zeta,\{y\}$} (Start);
\end{tikzpicture}
	\caption{A test purpose for the conveyor belt.}
	\label{fig:tp}
\end{wrapfigure}
\begin{example}
\label{ex:exemple_of_test_purpose}
Figure~\ref{fig:tp} is a test purpose for our conveyor-belt
example. We~want to be sure that it is possible to ship a package to
destination~2 in less than 5 time units, while avoiding to go in
waste. The $\Accept$ set is limited to a location, named
\Accept. We~note $\oth$ the set of transitions that reset no clocks, and
is enabled for an action other than $\zeta$ when no other transition
is possible for this action. This set serves to complete the test purpose. The~test purpose has a proper clock~$y$.
\end{example}

In practice, conformance testing links a mathematical model,
the~specification, and a black-box implementation, that is a real-life
\emph{physical} object observed by its interactions with the
environment.  In~order to formally reason about conformance, one~needs
to bridge the gap between the mathematical world and the physical world.
We~then assume that the implementation corresponds to an unknown~TAIO.
\begin{defn}{}{Implementation}
  Let $\mS = (L^\mS, l_0^\mS,
  (\Sigma_?\cup\{\zeta\})\uplus\Sigma_!\uplus\Sigma_\tau,
  X_p^\mS,I^\mS,E^\mS\cup \Restart)$ be a specification
  TAIO.  An \emph{implementation} of~$\mS$ is an input-complete and
  non-blocking TAIO $\mI = (L^\mI, l_0^\mI,
  (\Sigma_?\cup\{\zeta\})\uplus\Sigma_!\uplus\Sigma_\tau^\mI,
  X_p^\mI,I^\mI,E^\mI)$.
  We note $\mI(\mS)$ the set of possible implementations of~$\mS$.
\end{defn}

The hypotheses made on implementations are not restrictions, but model
real-world contingencies: the environment might always provide any
input and the system cannot alter the course of time.

Having defined the necessary objects, it is now possible to introduce
the \emph{timed input-output conformance}~($\tioco$)
relation~\cite{KT09}. Intuitively, it~can be understood as
"after any specified behaviour, outputs and delays of the
implementation should be specified".
\begin{defn}{}{tioco}
Let $\mS$ be a specification and $\mI\in \mI(\mS)$. We~say that~$\mI$
\emph{conforms~to}~$\mS$ for~$\tioco$, and write $\mI\tioco \mS$
when:
\[
\forall \sigma\in \Traces(\mS), \Sout(\mI \after 
\sigma)\subseteq \Sout(\mS \after  \sigma)
\]
\end{defn}
Note that it is not assumed that restarts are well implemented: if~they
are not, it~is significant only if it induces non-conformant
behaviours.


\subsection{Combining specifications and test purposes} 
\label{sub:combining_specification_and_test_purpose}

Now that the main objects are defined, we explain how the behaviours
targeted by the test purpose~$\TP$ are characterized on the
specification~$\mS$ by the construction of the product~OTAIO
$\mP=\mS\times\TP$.  Since $\mS$ is a TAIO and the observed clocks
of~$\TP$ are exactly the clocks of~$\mS$, the~product~$\mP$ is
actually a~TAIO.  Furthermore, since $\TP$ is complete, $\Seq(\mP) =
\Seq(\mS)$.
This entails that $\mI \tioco \mS$ is equivalent to
$\mI\tioco \mP$.  Note in particular that $\zeta$ of $\mS$
synchronize with $\zeta$ of $\TP$, which are available everywhere.

By defining accepting locations in the product by
$\Accept^{\mP}=L^\mS\times\Accept^\TP$, we~get that sequences
accepted in~$\mP$ are exactly sequences of~$\mS$ accepted by~$\TP$.

\begin{example}
\label{ex:exemple-prod}
Fig.~\ref{fig:product}
represents the product of the conveyor-belt specification of
Fig.~\ref{fig:conv_belt} and the test purpose of
Fig.~\ref{fig:tp}. All~nodes are named by the first letters of the
corresponding states of the specification (first) and of the test
purpose. The~only accepting location is~$(D_2,A)$.
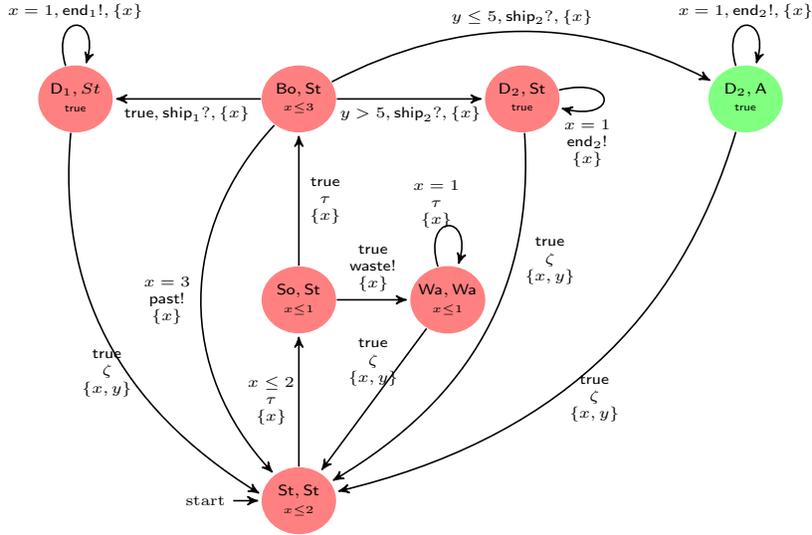
\begin{figure}[!h]
\centering
\begin{tikzpicture}[yscale=.9,->,>=stealth',shorten >=1pt,auto,node distance=3.3cm,semithick, scale=.9,  every node/.style={transform shape}]
  \tikzstyle{every state}=[fill=red!50,draw=none,text=black, font=\scriptsize, inner sep=0pt, minimum size=1.1cm]
  \tikzstyle{every path}=[font=\scriptsize, below]
  \node[state] (End1S)     [] {$\nodeinv{\aD_1,St}{\true}$};
  \node[state] (Boarding) [right of=End1S] {$\nodeinv{\Bo,\St}{x\leq 3}$};
  \node[state] (End2S)     [right of=Boarding] {$\nodeinv{\aD_2,\St}{\true}$};
  \node[state, fill=green!50] (End2A)     [right of=End2S] {$\nodeinv{\aD_2,\aA}{\true}$};
  \node[state] (Sort)     [below of=Boarding] {$\nodeinv{\So,\St}{x\leq 1}$};
  \node[state] (Waste)    [right of=Sort,node distance=2.2cm] {$\nodeinv{\Wa,\Wa}{x\leq 1}$};
  \node[state,initial] (Start)    [below of=Sort] {$\nodeinv{\St,\St}{x\leq 2}$};

\path (Start)     edge node[left=-1mm] {$\stack{x\leq2}{\tau}{\{x\}}$} (Sort)
      (Sort)      edge node[above] {$\stack{\true}{\waste!}{\{x\}}$} (Waste)
      (Sort)      edge node[right] {$\stack{\true}{\tau}{\{x\}}$}   (Boarding)
      (Boarding)  edge node {$\true,\ship_1?,\{x\}$} (End1S)
      (Boarding)  edge node {$y>5,\ship_2?,\{x\}$} (End2S)
      (Boarding)  edge [bend left] node[above] {$y\leq5,\ship_2?,\{x\}$} (End2A)
      (Boarding)  edge [bend right=40] node [left] {$\stack{x=3}{\past!}{\{x\}}$}   (Start)
      (End1S)     edge [loop above] node [above] {$x=1,\aend_1!,\{x\}$}  (End1S)
      (End2S)     edge [loop right] node[pos=.8,below]  {$\stack{x=1}{\aend_2!}{\{x\}}$}  (End2S)
      (End2A)     edge [loop above] node [above] {$x=1,\aend_2!,\{x\}$}  (End2A)
      (End1S)     edge [bend right] node {$\stack{\true}{\zeta}{\{x,y\}}$}  (Start.170)
      (End2S)     edge[bend left] node[pos=.3,right] {$\stack{\true}{\zeta}{\{x,y\}}$}  (Start)
      (End2A)     edge [bend left] node {$\stack{\true}{\zeta}{\{x,y\}}$}  (Start)
      (Waste)     edge [loop above] node [pos=.25,above] {$\stack{x=1}{\tau}{\{x\}}$}   (Waste)
      (Waste)     edge node[above=1mm] {$\stack{\true}{\zeta}{\{x,y\}}$}  (Start);
\end{tikzpicture}

	\caption{Product of the conveyor belt specification and the presented test purpose.}
	\label{fig:product}
\end{figure}
\end{example}

We make one final hypothesis: we consider only pairs of
specifications~$\mS$ and test purposes~$\TP$ whose product~$\mP$ can
be exactly determinized by the determinization game presented
in~\cite{BSJK15}. This restriction is necessary for technical
reasons: if~the determinization is approximated, we~cannot ensure
that restarts are still reachable in~general.
Notice that it is satisfied in several classes of
automata, such as strongly non-zeno automata, integer-reset automata,
or event-recording automata.
\par Given the product $\mP=\mS\times\TP$, let  $\mdp$ be its exact determinization.
 In~this case, $\Traces(\mdp)=\Traces(\mP)$,
hence the reachability of $\zeta$ transitions is preserved.  
Moreover the traces leading to $\Accept^\mdp$ and 
$\Accept^\mP$ are the same.

\begin{example}
\label{ex:exemple-DP}
The automaton in Fig. \ref{fig:DP}~is a deterministic approximation of the product presented in Fig. \ref{fig:product}. The internal transitions have collapsed, leading to an augmented $\Start$ locality.
\begin{figure}[!h]
\centering
\begin{tikzpicture}[yscale=.9,->,>=stealth',shorten >=1pt,auto,node distance=3.3cm,semithick, scale=.9,  every node/.style={transform shape}]
  \tikzstyle{every state}=[fill=red!50,draw=none,text=black, font=\scriptsize, inner sep=0pt, minimum size=1.1cm]
  \tikzstyle{every path}=[font=\scriptsize, below]
  \node[state] (End1S)     [] {$\nodeinv{\aD_1,St}{\true}$};
  \node[state] (Start)    [right of=End1S] {$\nodeinv{\St}{x\leq 6}$};
  \node[state] (End2S)     [right of=Start] {$\nodeinv{\aD_2,\St}{\true}$};
  \node[state, fill=green!50] (End2A)     [right of=End2S] {$\nodeinv{\aD_2,\aA}{\true}$};
  \node[state] (Waste)    [above of=Start,node distance=2.2cm] {$\nodeinv{\Wa}{\true}$};

\path (Start)     edge [bend left] node[left=-1mm] {$\stack{x\leq3}{\waste!}{\{x\}}$} (Waste)
      (Start)     edge node [above] {$\true,\ship_1?,\{x\}$} (End1S)
      (Start)     edge [bend left] node [below=1mm] {$y>5,\ship_2?,\{x\}$} (End2S)
      (Start)     edge [bend left] node [above] {$y\leq5,\ship_2?,\{x\}$} (End2A)
      (Start)     edge [loop below] node [below] {$3\leq x\leq6, \past!, \{x\}$}   (Start)
      (End1S)     edge [loop above] node [above] {$x=1,\aend_1!,\{x\}$}  (End1S)
      (End2S)     edge [loop right] node[right]  {$\stack{x=1}{\aend_2!}{\{x\}}$}  (End2S)
      (End2A)     edge [loop above] node [above] {$x=1,\aend_2!,\{x\}$}  (End2A)
      (End1S)     edge [bend right] node {$\true,\zeta,\{x,y\}$}  (Start)
      (End2S)     edge node[below] {$\true,\zeta,\{x,y\}$}  (Start)
      (End2A)     edge [bend left] node {$\true,\zeta,\{x,y\}$}  (Start)
      (Waste)     edge node[right=-1mm] {$\stack{\true}{\zeta}{\{x,y\}}$}  (Start);
\end{tikzpicture}

  \caption{A deterministic approximation of the product.}
  \label{fig:DP}
\end{figure}
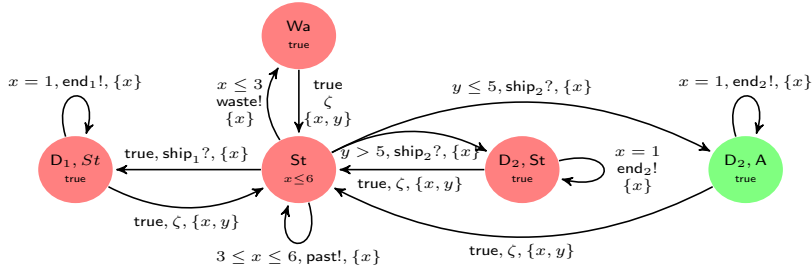
\end{example}



\subsection{Accounting for failure} 
\label{sub:accounting_for_failure}
At this stage of the process, we dispose of a deterministic and
fully-observable TAIO~$\mdp$ having exactly the same traces as the
original specification, and having a subset of its localities labelled
as accepting for the test purpose. From this TAIO, we aim to build a
tester, that can be able to monitor the implementation, feeding~it with
inputs and selecting verdicts from the returned outputs.

$\mdp$ models the accepted traces with $\Accept^\mdp$.
In~order to also explicitely model faulty behaviours (unspecified
outputs after a specified trace), we now
complete~$\mdp$ with respect to its output alphabet, by~adding an
explicit~$\Fail$ location. We call this completed TAIO the
\emph{objective-centered tester}.

\begin{defn}{}{Objective centered tester}
  Given a deterministic TAIO
  $\mdp = (L^\mdp, l_0^\mdp,
  \Sigma_?\uplus\Sigma_!\uplus\Sigma_\tau,\penalty1000
  X_p^\mdp,\penalty1000 I^\mdp,E^\mdp)$,
  we~construct its objective-centered tester
  $\ot = (L^\mdp\cup\{\Fail\},\penalty1000 l_0^\mdp,\penalty1000
  \Sigma_?\uplus\Sigma_!\uplus\Sigma_\tau,\penalty1000
  X_p^\mdp,I^\ot,E^\ot)$
  where $I^\ot(l)=\true$.
  The~set of transitions $E^\ot$ is defined from~$E^\mdp$~by:
  \begin{multline*}
    E^\ot = E^\mdp\cup
    \smash{\bigcup_{{\substack{l\in L^{\mdp}\\a\in\Sigma_!^{\mdp}}}}}
    \{(l,g,a,\emptyset,\Fail) \mid g\in \overline G_{a,l}\}
    \\{}\cup{} \{(\Fail,\true,a,\emptyset,\Fail) \mid a\in\Sigma^\mdp\}
    \end{multline*}
where for each~$a$ and~$l$, $\overline G_{a,l}$ is a set of guards
complementing the set of all valuations~$v$ for which an
$a$-transition is available from~$(l,v)$ (notice that $\overline
G_{a,l}$ generally is non-convex, so that it cannot be represented by
a single guard).

Verdicts are defined on the configurations of $\ot$ as follows: 
\begin{itemize}
\item $\VPass = \bigcup_{l\in \Accept^\mdp}(\{l\}\times I^\mdp(l))$,
\item $\VFail = \{\Fail\/\}\times\mathbb{R}_{\geq0}
    \cup\bigcup_{l\in L^\mdp}\left(\{l\}\times \bigl(\bbR_{\geq 0}^{X_p}\setminus I^\mdp(l)
    \bigr)\right)$.
\end{itemize}
\end{defn}

Notice that we do not define the usual \VInconc verdicts (\ie~configurations in which
we cannot conclude to non-conformance, nor accept the run with
respect to the test purpose)
as we will enforce the apparition of \VPass or~\VFail. \VPass~corresponds
to behaviours accepted by the test purpose, while \VFail corresponds
to non-conformant behaviours.
Note
that~$\ot$ inherits the interesting structural properties
of~$\mdp$. More importantly, $\zeta$ is always reachable as long as no
verdict has been emited, and $\ot$ is repeatedly-observable out
of~$\VFail$.

It remains to say that $\ot$ and $\mdp$ model the same
behaviours. Obviously, their~sets of traces differ,
but the traces added in~$\ot$ precisely correspond to
runs reaching~\VFail.
We~now define a specific subset
of runs, sequences and traces corresponding to traces that are
meant to be accepted by the specification.

\begin{defn}{}{ConformingExecutions}
A run~$\rho$ of an objective-centered tester~$\ot$ is said \emph{conformant} if
it does not reach~\VFail.
We~note $\Ex_{\conf}(\ot)$ the set of conformant runs of~$\ot$, and
$\Seq_{\conf}(\ot)$ (resp.~$\Traces_{\conf}(\ot)$)
the corresponding sequences (resp.~traces).
We note $\Ex_{\fail}(\ot)=\Ex(\ot)\setminus \Ex_{\conf}(\ot)$ and similarly
for the sequences and traces.
\end{defn}
The conformant traces are exactly those specified by~$\mdp$,
i.e. $\Traces(\mdp)=\Traces_{\conf}(\ot)$ and correspond to executions \tioco-conformant with the specification, while $\Ex_{\fail}$ are runs where a non-conformance is detected.

\section{Translating objectives into games} 
\label{sec:translating_objectives_into_games}
In this section, we interpret 
objective-centered tester into games between the tester and the implementation
and propose strategies that try to avoid control losses.
We then introduce a scope in which
the tester always has a winning strategy, and discuss the properties
of the resulting test cases (\ie~game structure and built strategy).

We~want to enforce conclusive verdicts when running test cases, i.e. either the
implementation does not conform to its specification (\VFail verdict)
or the awaited behaviour appears (\VPass verdict). We thus say that an
execution~$\rho$ is winning for the tester if it reaches a \VFail or \VPass configuration
and note~$\Win(\g)$ the set of such executions.  In~the following, we consider the TGA \(
\g^\ot = (L^\ot, l_0^\ot,\penalty1000
\Sigma_?^\ot\uplus\Sigma_!^\ot,\penalty1000 X_p,I^\ot,E^\ot) \) where
the controllable actions are the inputs $\Sigma_c=\Sigma_?^\ot$
and the uncontrollable actions are the outputs
$\Sigma_u=\Sigma_!^\ot$.

\subsection{Rank-lowering strategy} 
\label{sub:rank_lowering_strategy}

In this part, we restrict our discussion to TGAs where
\VPass configurations are reachable (when seen as plain
TAs). Indeed, if none can be reached,
and we will discuss the fact that the proposed method can detect this fact,
trying to construct a strategy seeking a \VPass verdict is hopeless.
This~is a natural restriction, as it only rules out unsatisfiable test purposes.

The tester cannot force the occurrence of a non-conformance
(as he does not control~outputs and delays),
and hence cannot push the system into a \VFail
configuration. A~strategy for the tester should thus target the
\VPass~set in a partially controlable way, while monitoring \VFail.
For that purpose, we  define a hierarchy of
configurations, depending on their "distance" to \VPass.
This uses a backward algorithm, for which we define the predecessors of
a configuration.

Given a set of configurations $S'\subseteq S$ of $\g^\ot$, we define
three kinds of predecessors, letting $\overline{V}$ denote the
complement of~$V$:
\begin{itemize}
\item discrete predecessors by a sub-alphabet $\Sigma'\subseteq \Sigma$:
  \[
  \Pred_{\Sigma'}(S')=\{(l,v) \mid \exists a\in\Sigma',\
    \exists (l,a,g,X',l')\in E, 
    v\models g \wedge (l',v_{[X'\leftarrow0]})\in S'\}	
  \]
\item timed predecessors, while avoiding a set $V$ of configurations:
  \[
  \tPred(S',V)=\{(l,v) \mid \exists \delta\in\mathbb{R}_{\geq0},\
  (l,v+\delta)\in S' \wedge
  \forall\ 0\leq \delta' \leq\delta.\ (l,v+\delta')\notin V \}
  \]
  We furthermore note $\tPred(S')=\tPred(S',\emptyset)$.
\item final timed predecessors are defined for convenience (see below):
  \[
  \ftPred(S') = \tPred(\VFail,\Pred_{\Sigma_u}(\overline{S'}))\cup
  \overline{\tPred(\Pred_\Sigma(\overline{S'}))}
	\]
\end{itemize}
The final timed predecessors correspond to
situations where the system is 'cornered', having the choice between
taking an uncontrollable transition to~$S'$ (as no uncontrollable
transition to~$\overline{S'}$ will be available) or reach
\VFail. Such situations are not considered as
control losses, as the system can only take a beneficial
transition for the tester (either by going to $S'$ or to~\VFail).
Note that $\tPred$ and $\ftPred$ need not return convex sets,
but are efficiently computable using $\Pred$ and simple set
constructions~\cite{CDF+05}. 
%
Now, using these notions of predecessors, a
hierarchy of configurations based on the 'distance' to~\VPass
is defined.
\begin{defn}{}{ConfigHierarchy}
The sequence $(W^j_i)_{j,i}$ of sets of configurations is defined~as: 
\begin{itemize}
\item $W^0_0 = \VPass$
\item $W^j_{i+1} = \pi(W^j_i)$ with $\pi(S')=\tPred\left(S'\cup
  \Pred_{\Sigma_c}(S'),\Pred_{\Sigma_u}(\overline{S'})\right) \cup
  \ftPred(S')$
\item \( W^{j+1}_0 = \tPred(W^j_{\infty}\cup
  \Pred_{\Sigma}(W^j_{\infty})) \) with $W^j_{\infty}$ the limit\footnote{The sequence~$(W_i^j)_i$
  is non-decreasing,
  and can be computed in terms of clock regions; hence the limit
  exists and is reached in a finite number of
  iterations~\cite{CDF+05}.}  of the sequence~$(W^j_i)_i$.
 
\end{itemize}
\end{defn}
In this hierarchy, $j$~corresponds to the minimal number of
\emph{control losses}
the tester has
to go through (in~the worst case) in order to reach~\VPass, and $i$ corresponds to the
minimal number of steps before the next control loss (or~to~$\VPass$).
The $W^{j+1}_0$ are considered 'control losses' as the implementation might take an output transition leading to an undesirable configuration (higher on the hierarchy). On the other hand, in the construction of $W^j_i$ the tester keep a full control, as it is not possible to reach such bad configuration with an incontrolable transition.
Notice that the sequence $(W^j_i)$ is an increasing sequence of regions, and hence can be computed in time exponential in~$X^\ot$ and linear in~$L^\ot$.

We then have the following property:

\begin{restatable}{proposition}{propcoverage}
\label{coverage}
There exists $i,j\in\mathbb{N}$ such
that $\Reach(\g^\ot)\setminus\VFail\subseteq W^j_i$.
\end{restatable}
As explained above, this property is based on the assumption
that the \VPass verdict is reachable.
Nevertheless, if~it is~not it will be detected during the hierarchy construction
that will converge to a fixpoint not including~$s_0^\g$.
As all the configurations in which we want to define a strategy are
covered by the hierarchy, we can use it to define a partial order.
\begin{defn}{}{Rank}
  Let $s\in \Reach(\g^\ot)\setminus\VFail$. The \emph{rank} of $s$~is:
  \[
  r(s) = (j_s=\argmin_{j\in\mathbb{N}}(s\in W^j_{\infty}), i_s =
  \argmin_{i\in\mathbb{N}}(s\in W^{j_s}_i))
  \]
\end{defn}
For $r(s)=(j,i)$; $j$ is the minimal number of control losses before reaching an accepting state, and $i$ is the minimal number of steps in the strategy before the next control loss.
We note $s\po s'$
when $r(s)\lo r(s')$, where $\lo$ is the lexical order on $\mathbb{N}^2$.
\begin{restatable}{proposition}{propPartialOrder}
$\po$ is a partial order on $\Reach(\g^\ot)\setminus\VFail$.
\end{restatable}

We dispose of a partial order on configurations, with $\VPass$ being the
minimal elements. We~use it to define a strategy trying to decrease
the rank during the execution.
For any~$s\in S$, we~write $r^{-}(s)$ for the largest rank such that
$r^{-}(s)<_{\mathbb{N}^2} r(s)$, and $W^{-}(s)$ for the associated set
in~$(W_i^j)_{j,i}$.
We~(partially) order pairs~$(\delta,a)\in\mathbb{R}_{\geq
0}\times\Sigma$ according to~$\delta$.
\begin{defn}{}{RankLowerStrat}
A strategy~$f$ for the tester is \emph{rank-lowering} if,
for any finite run~$\rho$ with $\last(\rho)=s=(l,v)$, 
it selects the lesser delay satisfying one of the following constraints:
\begin{itemize}
  \item if
  $s\in \tPred(\Pred_{\Sigma_c}(W^{-}(s)))$,
  then $f(\rho)=(\delta,a)$ with~$a\in\Sigma_c$ s.t. there exists
  $e\in E$ with $\act(e)=a$ and $s\xrightarrow{\delta}\xrightarrow{e}
  t$ with $t\in W^{-}(s)$, and $\delta$ is minimal in the following
  sense: if $s\xrightarrow{\delta'}\xrightarrow{e'} t'$ with $t'\in
  W^{-}(s)$ and $\delta'\leq\delta$, then $v+\delta$ and $v+\delta'$
  belong to the same region;

  \item if
  $s\in \tPred(W^{-}(s))$,
  then $f(\rho)=(\delta,\bot)$ such that $s\xrightarrow{\delta} t$ with
  $t\in W^{-}(s)$, and $\delta$ is minimal in the same sense as
  above;
 \item otherwise $f(\rho)=(\delta,\bot)$ where
  $\delta$ is maximal in the same sense as above (maximal
  delay-successor region).
\end{itemize}

\end{defn}
The two first cases follow the construction of the $W^j_i$ and propose the shortest behaviour leading to $W^{-}$.
The third case corresponds, either to a configuration of $\VPass$, where $W^{-}$ is undefined, or to a $\ftPred$.
Notice that (possibly several) rank-lowering strategies always
exist.

\begin{example}
\label{par:example_strategies_}
An example of a rank-lowering strategy on the automaton of
Fig.~\ref{fig:DP}~is: in~$(\aD_2,\aA)$, play~$\bot$
(as~$W^0_0$~has been reached); in~$\St$, play~$(0,\ship_2?)$;
in any other state, play~$(0,\zeta)$.
Note that Fig.~\ref{fig:DP} has not been completed in a
objective-centered tester. This~does not impact the strategies, as the
transition to the \fail states lead to a victory, but are not targeted
by the strategies.
\end{example}


It is worth noting that even in a more general setup where the models
are not equipped with $\zeta$-transitions, as~in~\cite{BSJK15},
rank-lowering strategies may still be useful: as~they~are defined on
the co-reachable set of $\Accept$, they~can still constitute test
cases, and the configurations where they are not defined are exactly
the configurations corresponding to a $\VFail$ verdict or to an
$\VInconc$ verdict, i.e., no~conclusions can be made since an accepting
configuration cannot be reached.

\subsection{Making rank-lowering strategies win} 
\label{sub:making_rank_lowering_strategies_win}
A rank-lowering strategy is generally not a winning strategy: it relies
on the implementation fairly exploring its different possibilities and
not repeatedly avoiding an enabled transition. In this section,
fair runs are introduced, and the rank-lowering strategies
are shown to be winning in this subset of the runs.

\begin{restatable}{lemma}{lemmakey}\label{lm:key}\label{lm:Akey}
If $\ot$ is repeatedly-observable, then
for all $\rho=((s_i,\gamma_i,s_{i+1}))_{i\in\mathbb{N}}\in \Ex(\g)$
ending with an infinite sequence of delays, we~have\footnote{In this expression, $\infof i\in\bbN, \ \phi(i)$ means that $\phi(i)$ is true for infinitely many integers.}
\[
\rho\in \Ex_{\fail}(\g)\vee \exists e\in E^\g,\
  \infof i\in\mathbb{N},e\in \enab(s_i).
\]
\end{restatable}
This lemma ensures that we cannot end in a situation where no
transitions can be taken, forcing the system to delay
indefinitely. It~will be used with the support of fairness. In order to
introduce our notion of fairness, we define the infinite support of a
run.
\begin{defn}{}{InfiniteSupport}
  Let  $\rho$ be an infinite run, its \emph{infinite support} 
 $\Inf(\rho)$ is the set of regions appearing infinitely often in $\rho$.
\begin{xalignat*}1
	\Inf((s_i,\gamma_i,s_{i+1})_{i\in\mathbb{N}})  &=\{r \mid \infof  i\in\mathbb{N},\ s_i\in r \vee {} \\
        &\qquad\qquad
  (\gamma_i\in\mathbb{R}_{\geq0}\wedge \exists s'_i\in r,\ \exists \delta_i < \gamma_i,\ s_i\xrightarrow{\delta_i}s'_i)\}
\end{xalignat*}
\end{defn}

The notion of enabled transitions and delay transitions are extended
to regions as follows: for a region $r$, we~let $\enab(r)=\enab(s)$
for any~$s$ in~$r$, and write $r\xrightarrow{\mathtt{t}} r'$ for all
time-successor region~$r'$ of~$r$.

\begin{defn}{}{FairRun}
An infinite run $\rho$ in a TGA
$\g=(L,l_0,\Sigma_u\uplus\Sigma_c,X,I,E)$ (with timed transitions system
$\mT = (S,s_0,\Gamma,\mathord\rightarrow_\mT)$) is said to be \emph{fair} when:
\begin{align*}
&\forall e\in E, (\act(e)\in\Sigma_u\Rightarrow (\exists r \in \Inf(\rho), r\xrightarrow{e}r'\Rightarrow r'\in \Inf(\rho)))
\ \wedge{}
\\
&\forall r\in \Inf(\rho), \exists \gamma\in (\enab(r)\cap\{e \mid \act(e)\in\Sigma_c\})\cup
\{\mathtt{t}\}
), r\xrightarrow{\gamma}r' \wedge r'\in \Inf(\rho)
\end{align*}
We note $\Fair(\g)$ the set of fair runs of $\g$.
\end{defn}
Fair runs model restrictions on the system runs corresponding to
strategies of the system. The first part of the definition assures that any infinitely-enabled action of the implementation will be taken infinitely many times, while the second part ensures that the implementation
will infinitely often let the tester play, by ensuring that a delay or controlable action will be performed.
It~matches the "strong fairness" notion used in
model checking.  Restricting to fair runs is sufficient to ensure a winning
execution when the tester uses a rank-lowering strategy. Intuitively,
combined with Lemma~\ref{lm:key} and the repeated-observability assumption,
it assures that the system will keep providing outputs until a verdict
is reached, and allows to show the following property.

\begin{restatable}{proposition}{propwin}\label{prop:win}\label{prop:Awin}
Rank-lowering strategies are winning on~$\Fair(\g)$ (i.e., all fair
outcomes are winning).
\end{restatable}
Under the hypothesis of a fair implementation, we thus have identified
a test-case generation method, starting from the specification with
restarts and the test purpose, and constructing a test case as a
winning strategy on the game created from the objective-centered
tester. The complexity of this method is exponential in the size of $\mdp$. 
More precisely: 
\begin{restatable}{proposition}{complexity}\label{prop:com}\label{prop:Acom}
Given a deterministic product $\mdp$, $\ot$ can be linearly computed from $\mdp$, and the construction of a strategy relies on the construction of the $W^j_i$ and is hence exponential in $X^\mdp$ and linear in $L^\mdp$.
\end{restatable}
Note that if $\mdp$ is obtained from $\mP$ by the game presented in \cite{BSJK15}, then $L^\mdp$ is doubly exponential in the size of $X^\mS\uplus X^\TP \uplus X^\mdp$ (notice that in the setting of~\cite{BSJK15}, $X^\mdp$~is a parameter of the algorithm).

\subsection{Properties of the test cases} 
\label{sub:properties_of_the_test_cases}
Having constructed strategies for the tester, and identified a scope
of implementation behaviours that allows these strategies to enforce a
conclusive verdict, we now study the properties obtained by the
test generation method presented above. We~call \emph{test case} a
pair~$(\g,f)$ where $\g$ is the game corresponding to the
objective-centered tester $\ot$, and $f$ is a rank-lowering strategy on~$\g$.
We~note~$\TC(\mS,\TP)$ the set of possibles test cases generated from a specification~$\mS$ and a test purpose~$\TP$, and~$\TC(\mS)$ the test cases for any test purpose.
Recall that it is assumed that the test purposes associated with a specification are restricted to those leading to a determnizable product.
 \emph{Behaviours} are defined as the possible outcomes of a test case combined with an
implementation, and model their parallel composition.
\begin{defn}{}{Behaviour}
Given a test case $(\g,f)$ and an implementation $\mI$,
their \emph{behaviours} are the runs
$((s_i,s'_i),(e_i,e'_i),(s_{i+1},s'_{i+1}))_i$ such that
$((s_i,e_i,s_{i+1}))_i$ is an outcome of~$(\g,f)$,
$((s'_i,e'_i,s'_{i+1}))_i$ is a run of~$\mI$, and for all~$i$, either
$e_i=e'_i$ if $e_i\in\mathbb{R}_{\geq 0}$ or $\act(e_i)=\act(e'_i)$
otherwise.  We~write $\Behaviour(\g,f,\mI)$ for the set of behaviours
of the test case~$(\g,f)$ and of the implementation~$\mI$.
\end{defn}

We say that an implementation~$\mI$ fails a
test case $(\g,f)$, and note $\mI \fails (\g,f)$, when there exists
a run in $\Behaviour(\g,f,\mI)$ that reaches~\VFail.
Our~method is sound, that~is, a~conformant implementation cannot be detected as faulty.
\begin{restatable}{proposition}{propSoundness}\label{prop:Soundness}
The test-case generation method is \emph{sound}:
for any specification~$\mS$, it~holds
\[
\forall \mI\in\mI(\mS),\
\forall(\g,f)\in\TC(\mS),\ (\mI \fails (\g,f)\Rightarrow
\neg(\mI \tioco \mS)).
\]
\end{restatable}
The proofs of this property and the following one are based on the exact correspondance between \VFail and the faulty behaviours of $\mS$, and use the trace equivalence of the different models ($\mdp$, $\mP$ and $\mS$) to conclude. As they exploit mainly the game structure, fairness is not used.
\par We define the notion of outputs after a trace in a behaviour, allowing
to extend $\tioco$ to these objects and to state a strictness
property. Intuitively, when a non-conformance appears it should be detected.
\begin{defn}{}{OutputsAfterBehaviour}
Given a test case $(\g,f)$ and an implementation $\mI$, for a trace~$\sigma$:
\begin{multline*}
  \Sout(\Behaviour(\g,f,\mI)\after\sigma)=\\
  \{a\in \Sigma_!\cup\mathbb{R}_{\geq0}\mid \exists \rho\in
  \Behaviour(\g,f,\mI),\ \Trace(\rho)=\sigma\cdot a\}
\end{multline*}
\end{defn}

\begin{restatable}{proposition}{propStrictness}
The test generation method is \emph{strict}:
given a specification~$\mS$,
\[
\forall \mI\in\mI(\mS),\ \forall(\g,f)\in\TC(\mS),\ 
\neg(\Behaviour(\g,f,\mI)\tioco \mS)\Rightarrow \mI \fails (\g,f)
\]
\end{restatable}

This method also enjoys a precision property: 
traces leading the test case to \VPass are exactly traces conforming 
to the specification and accepted by the test purpose.
The proof of this property uses the exact encoding of the $\Accept$ states and the definition of $\VPass$. As the previous two, it then propagates the property through the different test artifacts.
\begin{restatable}{proposition}{propPrecision}
The test case generation method is \emph{precise}:
for any specification~$\mS$ and test purpose~$\TP$ it can be stated that
\begin{multline*}
  \forall (\g,f)\in\TC(\mS,\TP),\forall \sigma\in \Traces(\Outcome(s_0^\g,f)), 
  \\\g\after \sigma \in \VPass
  \ \Leftrightarrow \
	(\sigma\in \Traces(\mS)\wedge \TP \after \sigma\cap \Accept^\TP \neq \emptyset)
\end{multline*}
\end{restatable}

Lastly, this method is exhaustive in the sense that for any non-conformance, there exist a test case that allows to detect it, under fairness assumption.

\begin{restatable}{proposition}{propExhaust}
The test generation method is \emph{exhaustive}:
for any exactly determinizable specification $\mS$
and any implementation $\mI\in\mI(\mS)$ making fair runs 
\[
\neg(\mI\tioco\ \mS)\Rightarrow \exists (\g,f)\in \TC(\mS), \mI \fails
(\g,f).
\]
\end{restatable}
To demonstrate this property, a test purpose is tailored to detect a given non-conformance, by targeting a related conformant trace.


\section{Conclusion} 
\label{sec:conclusion}

This paper proposes a game approach to the controllability problem for conformance testing from timed automata (TA) specifications.
It defines a test synthesis method that produces test cases whose aim is to maximize their control upon the implementation under test, while detecting non-conformance.
Test cases are defined as strategies of a game between the tester and the implementation, based on the distance to the satisfaction of a test purpose, both in terms of number of transitions and potential control losses. 
Fairness assumptions are used to make those strategies winning and are proved sufficient to obtain the exhaustiveness of the test synthesis method, together with soundness, strictness and precision.

This paper opens numerous directions for future work.
First, we  intend to tackle partial observation in a more complete  and practical way.
One direction consists in finding weaker conditions
under which approximate  determinization~\cite{BSJK15} preserves strong connectivity, a condition for the existence of winning strategies.
One could also consider a mixture of our model and the model of~\cite{DLLMN10} whose observer predicates are clearly adequate in some contexts.
Quantitative aspects could also better meet practical needs.
The distance to the goal could also include the time distance or costs of transitions,
in particular to avoid restarts when they induce heavy costs but longer and cheaper paths are possible. 
The fairness assumption could also be refined.
For now it is assumed on both the specification and the implementation.
If the implementation does not implement some outputs, a tester could detect it with a bounded fairness assumption~\cite{Ram98}, adapted to the timed context (after sufficiently many experiments traversing some region all outputs have been observed), thus allowing a stronger conformance relation with egality of output sets.
A natural extension could also be to complete the approach in  a stochastic view. 
Finally, we plan to implement the results of this work in an open tool for the analysis of timed automata, experiment on real examples and check the scalability of the method.

\bibliographystyle{alpha}
\bibliography{biblio}

\ifarxiv
\clearpage
\appendix
\setcounter{theorem}{0}
\def\thetheorem{\thesection.\arabic{theorem}}
\def\theproposition{\thesection.\arabic{theorem}}
\def\thelemma{\thesection.\arabic{theorem}}
\def\thecorollary{\thesection.\arabic{theorem}}

\section*{Appendix}

We conduct here the proofs of the different claims. They are separated
according to the corresponding parts of the article: testing framework,
games, and test-case properties.

\section{Test framework} 
\label{sub:test_framework}
First, we prove the claim on the strong-connectivity of the reachable
part of a specification semantics. This property is made quite
intuitive by the introduction of always-reachable $\zeta$-transitions. 
As it grounds our approach, we still provide a formal
proof.
\propSC*
\begin{proof}
Let $s$ be a configuration of $\mT_\mS$ reachable from $s_0^\mS$.
By hypothesis, there exists a finite partial execution starting in $s$ which trace contains $\zeta$.
This trace leads to the configuration $s_0^\mS$
hence any reachable configuration of $\mT_\mS$ is reachable from $s$,
and we conclude that the reachable part of $\mT_\mS$ is strongly-connected. 
\end{proof}

We now prove some properties of the product between specifications and
test purposes.
\begin{prop}{}{Asec_prev}
Let $\mS$ be a specification and $\TP$ a test purpose on this specification.
Then
\[
	\Seq(\mS\times\TP) = \Seq(\mS)
\]
\end{prop}
\begin{proof}
It suffices to note that the set of sequences of the product of two
OTAIOs is the intersection of the sequences of the two original
OTAIOs~\cite{BSJK15}, so that $\Seq(\mS\times\TP)
= \Seq(\mS)\cap \Seq(\TP)$; we also note that $\TP$ is complete, and hence
it~accepts any sequence. 
\end{proof}

By~projection on traces, we immediately get:
\begin{cor}{}{Atr_prev}
Let $\mS$ be a specification, and $\TP$ a test purpose. $\mS$ and $\mS\times\TP$ are trace-equivalent.
\end{cor}
\begin{cor}{Strong connectivity of a product semantic}{}
Let $\mS$ be a specification, $\TP$ a test purpose and
$\mT_{\mS\times\TP}$ its associated timed transition system. The
reachable part of $\mT_{\mS\times\TP}$ is strongly-connected.
\end{cor}
\begin{proof}
This proof is derived from the proof of Prop.~\ref{pr:Asc_s}. Although
they are really close, we have to do it again as the product does not
ensure any relation on the semantics.

Let $((l^1,l^2),v)$ be a reachable configuration of
$\mT_{\mS\times\TP}$. There exists a finite partial execution starting
in $(l^1,v)$ whose trace contains~$\zeta$, and thus by
Corollary~\ref{cr:Atr_prev} there exists a finite partial execution
starting in $((l^1,l^2),v)$ whose trace contains~$\zeta$.  Hence this
transition leads to the configuration
$s_0=(l_0^{\mS\times\TP},\overline{0})$. It~comes that there exists a
finite partial execution from $((l^1,l^2),v)$ to~$s_0$. Hence any
reachable configuration of $\mT_{\mS\times\TP}$ is reachable from
$((l^1,l^2),v)$. It~can then be concluded that the reachable part of
$\mT_\mS$ is strongly-connected.
\end{proof}

The following properties concern the objective-centered
tester. They~mainly amount to proving that the objective-centered
tester keeps the interesting properties of the product, and that its
traces are related to those of the previous automata.
\begin{prop}{}{}
Let $\mdp$ be the exact determinization of the product~$\mP$ between a
specification and a test purpose, and $\ot$ its associated
objective-centered tester.
Then
\[
	\Traces(\mdp)=\Traces_{\conf}(\ot).
\]
\end{prop}

\begin{proof}
  An execution is in $\Ex_{\conf}(\ot)$ if it avoids the
  verdict~\VFail. This amounts to avoiding the location~$\Fail$ and
  respecting the invariants of~$\mdp$.  By~construction of~$\ot$, this
  corresponds exactly to the runs of~$\mdp$.
\end{proof}

\begin{lem}{}{Asep}
Given an objective-centered tester~$\ot$, we~have
\[
\Traces_{\conf}(\ot)\cap \Traces(\Ex_{\fail}(\ot))=\emptyset.
\]
\end{lem}

\begin{proof}
Let $\rho\in \Ex_{\fail}(\ot)$. Consider the longest prefix~$\rho'$
of~$\rho$ that does not reach~\VFail, and let~$e$ be the transition taken
after~$\rho'$ in~$\rho$. Two~cases should be considered:
\begin{itemize}
  \item If $e$ is a delay transition, then it violates the invariant of
    the location in~$\mdp$. By~determinism of~$\mdp$, $\mdp\after
    \Trace(\rho')$ is a singleton. Hence
    the same delay is not available after~$\rho$ in~$\mdp$.
  \item If $\act(e)\in\Sigma_!$ then this output is not specified in the
    current location of~$\mdp$. By~determinism of~$\mdp$, $\mdp\after
    \Trace(\rho')$ is a singleton. Hence
    transition~$e$ is not possible after~$\rho$ in~$\mdp$.
\end{itemize}
In both cases $\Trace(\rho)\notin \Traces_{\conf}(\ot)$. As this holds
for any run of $\Ex_{\fail}(\ot)$ we have the desired property.
\end{proof}

\begin{lem}{}{Ar_ot}
Let $\ot$ be an objective-centered tester. For~any location~$l$ in
$L^\ot\setminus\{\Fail\}$ there exists a finite partial execution
$\rho\in \pEx(\ot)$ starting in $l$ and containing a $\zeta$.
\end{lem}
\begin{proof}
The same result holds from any state in~$\mS$, and 
$\Traces(\mS)=\Traces(\mS\times\TP)$.
The~result follows by exact determinizability of the product~$\mS\times\TP$.
\end{proof}

This lemma is the reason why our method assumes exact
determinizability, as we can't ensure in general that the restart will
remain reachable: if determinization is approximated, some traces
might be lost.
\begin{cor}{}{Asc_ot}
Let $\ot$ be an objective-centered tester and $\mT^\ot$ its associated
timed transition system. Then $\Reach(\mT^\ot)\setminus\VFail$ is
strongly-connected.
\end{cor}

\begin{proof}
	The proof is the same as the one of
	Prop.~\ref{pr:Asc_s}, using Lemma~\ref{lm:Ar_ot}.
\end{proof}

\begin{lem}{Repeatedly-observable tester}{}
For a repeatedly-observable specification~$\mS$,
$\Reach(\ot)\setminus\VFail$ is repeatedly-observable.
\end{lem}

\begin{proof}
We know that $\Traces(\mS)=\Traces(\mP)$, hence for all $\sigma\in
\Traces(\mP)$, $\Sout(\mP\after\sigma) = \Sout(\mS\after \sigma)$. As
$\Traces(\mdp)=\Traces(\mP)$ by assumption, we~also know that
for all $\sigma\in \Traces(\mdp)$, $\Sout(\mP\after \sigma) \subseteq
\Sout(\mdp\after\sigma)$. It~comes 
\[
\forall\sigma\in \Traces(\ot),\ \Sout(\mS \after \sigma) \subseteq
\Sout(\ot\after \sigma)
\]
as $\ot$ only adds traces to~$\mdp$. Hence for all $s\in
\Reach(S^\ot)\setminus\VFail$, there exists $\mu\in \Seq(\ot)$ s.t.
$s\xrightarrow{\mu}\wedge
\Trace(\mu)\notin\mathbb{R}_{\geq0}$. Indeed, there exists $\sigma\in
\Traces_{\conf}(\ot)$ such that $s = \ot \after \sigma$ (as~$\ot$ is
deterministic outside of~\VFail) and for $s'\in \mS\after \sigma$,
there exists $\mu'$ such that $s'\xrightarrow{\mu'}\wedge
\Trace(\mu')\notin \mathbb{R}_{\geq0}$. It~suffices to take $\mu
\in\pSeq(\ot)$ such that $\Trace(\mu)=\Trace(\mu')$, and by the
previous trace-inclusion property, such a trace exists.
\end{proof}


\section{Games} 
\label{sub:game}
In this part the previous propositions are used to ensure that the sequence 
$(W^j_i)_{i,j}$ defines a partial order, and covers the reachable part of our
game.
\propcoverage*

\begin{proof}
Let $s \in \Reach(\g)\setminus\VFail$ be a reachable
configuration.
Since i)~$\VPass$ is
reachable from~$s_0^\mT$ (by~hypothesis); ii)~there is a path
from~$s$ back to the initial configuration 
(Corollary~\ref{cr:Asc_ot}),  
then \VPass is reachable from~$s$.
Moreover, there is such a path with
length bounded by the number of regions.

For each~$s\in\Reach(\g)\setminus\VFail$, we~fix a finite path
to~\VPass, and reason by induction on the length~$n$ of this path in
order to
prove that $s\in W^n_0$:
\begin{itemize}
\item Case $n=0$: in this case $s\in \VPass=W^0_0$
\item Inductive case: we assume that the result holds for~$n$, and
  take~$s$ with a path to~$\VPass$ of length~$n+1$.
  Then $s\xrightarrow{e} s'$ for some~$e$ with $\act(e)\in\Gamma$,
  and there is a path from~$s'$ to~$\VPass$ of length at most~$n$, so that
  $s'\in W^n_0$. Hence in the worst case $s\in W_0^{n+1}$.
\end{itemize}
This proves our result.
\end{proof}

\propPartialOrder*

\begin{proof}
$\po$ is an order because it directly inherits the properties of
  $\lo$. It is not total because several configurations can have the
  same rank.
\end{proof}

The following lemma is the key allowing to ensure that rank-lowering
strategies are winning on fair executions.

\bgroup\def\footnote#1{}
\lemmakey*
\egroup

\begin{proof}
  We show this lemma by contradiction. Assume that for some $\rho\in \Ex(\g)$,
  we have
  \[
	\rho\notin \Ex_{\fail}(\g)\wedge \forall e\in E^\g,\ \bfn
        i\in\mathbb{N},\ e\notin \enab(s_i).
\]
Let $\rho_{\max}$ be the shortest prefix such that no transition is
enabled after this prefix along~$\rho$ (it~exists because $E$ is
finite and there is only a finite number of these prefixes per element
of~$E$).
Consider any prefix~$\rho'$ of $\rho$ strictly containing~$\rho_{\max}$;
there is no partial sequence~$\mu$ such that
$\last(\rho')\xrightarrow{\mu}{}$ and
$\Trace(\mu)\notin\mathbb{R}_{\geq0}$, as there is no time successor
of $\last(\rho')$ with an enabled transition.  This contradicts the
repeated-observability of $\ot$ out of~$\Fail$ (as~$\g$ and $\ot$ are the
same automaton).
\end{proof}

\propwin*
  In order to make this proof, we reason on regions. For this purpose
  we extend to regions the notions of executions and enabled
  transitions. We furthermore note that a region is included in any
  $W^j_i$ it intersects.
We first prove the following lemma: 
\begin{lem}{}{delay}
Let $\rho\in\Fair(\g)$ be a fair execution and
$\reg\in\Inf(\rho)$. For any prefix~$\nu$ of~$\rho$ ending in~$\reg$,
and any rank-lowering strategy~$f$, noting $f(\nu)=(\delta,a)$, we~have
$\nu\xrightarrow{\delta}s\in\reg'$ and $\reg'\in\Inf(\rho)$.
\end{lem}
\begin{proof}
  By definition of a rank-lowering strategy, $\delta$ is a possible
  delay after $\last(\nu)$. Hence there exists $s$ and $\reg'$ such
  that $\nu\xrightarrow{\delta}s\in\reg'$. By~definition of
  rank-lowering strategies, it is always the same $\reg'$ for
  each~$\nu$.
  If $\reg=\reg'$ then we have our result. Otherwise, by~definition of
  outcomes, there is no transition labelled with a controllable
  transition leaving~$\reg$, and by our fairness assumption, there
  exists a (strict) time successor $\reg\xrightarrow{t}\reg''$
  of~$\reg$ such that $\reg''\in \Inf(\rho)$.
  If $\reg''=\reg'$ we have our result; otherwise, by~defintion of
  outcomes, $\reg''$~is a time predecessor of~$\reg'$, and applying the
  same arguments to $\reg''$ will create an induction (as~by
  definition of rank-lowering strategies, the~strategy after going
  from~$\nu$ to $\reg''$ is to delay to~$\reg'$). As~there is only
  finitely-many regions between~$\reg$ and~$\reg'$, we~have our result
  by the induction principle.
\end{proof}

With this lemma, the proof of the main proposition is made easier.
\begin{proof}[Proof (of Prop.~\ref{prop:Awin})]
  Let $\mT = (S,s_0,\Gamma,\mathord\rightarrow_\mT)$ be the timed
  transition system associated with
  $\g=(L,l_0,\Sigma_u\uplus\Sigma_c,X,I,E)$.  Let~$f$ be a
  rank-lowering strategy. We~want to prove that $\Outcome(s_0,f)\cap
  \Fair(\g)\subseteq \Win(\g)$. We~proceed by contradiction.
  
  Suppose there exists an infinite run $\rho\in \Outcome(s_0,f)\cap
  \Fair(\g)$ such that $\rho\notin \Win(\g)$. We denote $r_{\min}$ the
  minimal rank obtained in~$\Inf(\rho)$ and $\reg\in \Inf(\rho)$ such
  that $r(\reg)=r_{\min}$.  For~each prefix~$\nu$ of~$\rho$ ending in
  a configuration~$s=(l,v)\in\reg$, we~let
  $(\delta_\nu,a_\nu)=f(\nu)$. We~consider three cases,
  following the definition of rank-lowering strategies:
  \begin{itemize}
  \item 
    Assume
    $\last(\nu)\in\tPred(\Pred_{\Sigma_c}(W^{-}(\last(\nu))))$ and
    $a_\nu\in\Sigma_c$. By~Lemma~\ref{lm:delay}, there exists
    $\reg'\in\Inf(\rho)$ such that
    $\nu\xrightarrow{\delta_\nu}s_\nu$ and $s_\nu\in \reg'$. Hence
    there exists a transition $e$ such that $\act(e)=a_\nu$. The
    system cannot delay more by definition of $\Outcome(s_0,f)$ and by
    fairness, $\reg'\xrightarrow{e}\reg''$ and $\reg''\in
    W^{-}(\last(\nu))\cap\Inf(\rho)$, which contradicts the
    minimality of $r_{\min}$.
  \item 
    Assume $\last(\nu)\in\tPred(W^{-}(\last(\nu)))$ and
    $a_\nu=\bot$.  By~Lemma~\ref{lm:delay}, there exists
    $\reg'\in\Inf(\rho)$ such that
    $\nu\xrightarrow{\delta_\nu}s_\nu$ and $s_\nu\in \reg'$. By
    definition of the case, $\reg'\in W^{-}(\last(\nu))$, thus the
    minimality of $r_{\min}$ is contradicted.
  \item 
    We finally consider the last case:     as the
    duration in this one is maximal, we~have
    $\last(\nu)\notin\tPred(W^{-}(\last(\nu))
    \cup\Pred_{\Sigma_c}(W^{-}(\last(\nu))))$
    and there are two cases to consider:
    \begin{itemize}
    \item if $(W^{-}(\last(\nu))$ is undefined, then
      $r(\last(\rho))=(0,0)$ and $\rho$ is winning, which is a
      contradiction;
    \item otherwise, by~Proposition~\ref{coverage},
      $\last(\nu)\in\tPred(\Pred_{\Sigma_u}(W^{-}(\last(\nu))))$
      (this corresponds either to $\ftPred$ or to a
      $W^{j+1}_0$). Furthermore $a_\nu=\bot$ and $\delta_\nu$ leads
      to the maximal delay successor region, which we
      call~$\reg'$. By~Lemma~\ref{lm:delay},
      $\reg'\in\Inf(\rho)$. Hence, by~definition of~$\Inf$, all
      regions between $\reg$ and~$\reg'$ are in $\Inf(\rho)$. In
      particular, a~region $\reg''\in\Inf(\rho)$ such that
      $\reg''\xrightarrow{e}\reg'''\in W^{-}(\last(\rho))$ and
      $\act(e)\in\Sigma_u$ exists by definition of the case. Hence by
      fairness, $\reg'''\in\Inf(\rho)$ and the minimality of
      $r_{\min}$ is contradicted.\qed
    \end{itemize}
  \end{itemize}

\end{proof}


\section{Test case properties} 
\label{sub:test_case_properties}
\propSoundness*

\begin{proof}
Let $\mS$ be a specification, $\mI\in\mI(\mS)$ and $(\g,f)\in\TC(\mS)$. Suppose that $\mI \fails (\g,f)$, we will prove that $\neg(\mI \tioco \mS)$. 

Since  $\mI \fails (\g,f)$, there is a finite run $\rho$ of $\Behaviour(\g,f,I)$ such that $last(\rho)\in \VFail\times S^\mI$ and it is the first
configuration of $\rho$ in this set. Let $\sigma=\Trace(\rho)$. By
construction of \VFail, either $\sigma = \sigma'\cdot \delta$ (if
the configuration of \VFail reached corresponds to a faulty
invariant) or $\sigma = \sigma'\cdot a$ with $a\in \Sigma_!$ (and $\Fail$
is reached). In both cases $\Sout(\mI\after \sigma')\nsubseteq
\Sout(\mdp\after \sigma')$, and by definition
$\neg(I\tioco \mdp)$.

As $\Traces(\mP)=\Traces(\mdp)$ by exact-determinizability hypothesis,
${\neg(I\tioco \mP)}$. Finally, as $\Traces(\mP)=\Traces(\mS)$, we have
$\neg(I\tioco \mS)$, which concludes the proof.
\end{proof}
\begin{remark}
Note that this proof is more general than the property, as it does not
rely on the strategy. It hence proves the property for any
strategy~$f$ and not only for rank-lowering ones. The~key reason lies
in fact in the structure of~$\g$, and ensures that any run reaching
\VFail has the correct form.
\end{remark}

\propStrictness*

\begin{proof}
Let $\mS$ be a specification, $\mI\in\mI(\mS)$ and
$(\g,f)\in\TC(\mS)$. Suppose that
$\neg(\Behaviour(\g,f,\mI)\tioco \mS)$. We want to show that $\mI
\fails (\g,f)$. By definition of
$\neg(\Behaviour(\g,f,\mI)\tioco \mS)$, there exist $\sigma\in
\Traces(\mS)$ and
\[
a\in \Sout(\Behaviour(\g,f,\mI)\after \sigma) \setminus \Sout(\mS\after \sigma)
\]
Since
$\mdp$ is an exact determinization of $\mP$ we have the following
equalities:
$\Traces(\mS)=\Traces(\mP)=\Traces(\mdp)=\Traces_{\conf}(\ot)$.
Since $a\in \mathbb{R}_{\geq0}\cup \Sigma_!$,
$\sigma\cdot a\in \Traces(\ot)$
as~invariants have been
removed, and the automaton has been completed on $\Sigma_!$ with
transitions to $\Fail$). Hence $\sigma\cdot a \in
\Traces(\Ex_{\fail}(\ot))$. Thus, for $\rho\in \Behaviour(\g,f,\mI)$ such
that $\Trace(\rho)=\sigma\cdot a$, $\last(\rho)\in \VFail$ and $\mI
\fails (\g,f)$.
\end{proof}
Note that once again, the properties of the strategy are not used.

\propPrecision*

\begin{proof}
Let $\sigma$ be in $\Traces(\Outcome(s_0^\g,f))$. Then $\g\after \sigma \in
\VPass$ if, and only~if, the run~$\rho$ such that $\Trace(\rho)=\sigma$
(which is unique by determinism of~$\g$ outside~\VFail) is such that
$\last(\rho)\in \VPass$, \ie\ $\rho\in \Ex(\mdp)$ and
$\last(\rho)\in \Accept^\mdp$. Hence $\mdp \after  \sigma \in
\Accept^\mdp$ and as the determinization is exact, $\sigma\in
\Traces(\mP)$ and $\mP\after  \sigma\in \Accept^\mP$, which
gives by definition $ \sigma\in \Traces(\mS)\wedge \TP \after 
\sigma\cap \Accept^\TP \neq \emptyset$.
\end{proof}
The proof uses only properties of the game, and once more does not rely on the precise strategy used.

\propExhaust*

\begin{proof}
Let $\mS$ be a specification, and $\mI\in\mI(\mS)$ a non-conformant
implementation. By definition of $\neg(\mI\tioco \mS)$, there exists
$\sigma \in \Traces(\mS)$ and $a\in\mathbb{R}_{\geq0}\cup\Sigma_!$ such
that $a\in \Sout(\mI\after \sigma)$ and $a\notin
\Sout(\mS\after \sigma)$. As~$\mS$~is repeatedly-observable,
there exists $\delta\in\mathbb{R}_{\geq0}$ and $b\in \Sigma^\mS_{\obs}$
such that $\sigma\cdot \delta\cdot b\in \Traces(\mS)$. Because $\mS$ is also
non-blocking, if $a$ is a delay, we can take $b\in\Sigma_!^\mS$. Indeed,
otherwise there would be no trace controlled by the implementation for any
finite time (say, for time~$a$).

It is possible to build a test purpose $\TP$ that accepts exactly the
trace $\sigma\cdot\delta\cdot b$. It suffices to send every transition
that is not part of this trace to a sink location. As
$\sigma\cdot\delta\cdot b\in \Traces(\mS)$ it is also a trace of the product $\mP=\mS\times\TP$.
As $\mS$ is exactly determinizable and $\TP$ is deterministic, $\mP$ is exactly determinizable
by allowing enough resources to $\mdp$.
We thus obtain $\Traces(\mdp)=\Traces(\mP)$ and $\sigma\cdot \delta\cdot b\in \Traces(\mdp)$.
Hence, the minimal elements of
$\VPass$ are $\ot\after \sigma\cdot\delta\cdot b$.

From $\ot$ a test case $(\g,f)$ can be built, with $f$ a rank-lowering
strategy. By assumption, the implementation is playing fair
runs, hence $f$ is winning. So there exists $\rho\in
\Outcome(s_0^\g,f)$ such that $\Trace(\rho)=\sigma\cdot \delta\cdot b$, and thus
there exists $\rho'\in \Outcome(s_0^\g,f)$ such that
$\Trace(\rho')=\sigma$. By assumption, $\sigma\cdot a\in \Traces(\mI)$, and
depending on the nature of~$a$:
\begin{itemize}
  \item If $a\in\Sigma_!$ then $\sigma\cdot a\in \Outcome(s_0^\g,f)$ as $\g$
    is complete on $\Sigma_!$. Hence $\sigma\cdot a\in \Behaviour(\g,f,\mI)$
    and as $\sigma\cdot a\notin \Traces(\mS)$ and the determinization is
    exact, $\sigma\cdot a\notin \Traces_{\conf}(\ot)$ and $\g \after
    \sigma\cdot a\in \VFail$. Hence $\mI \fails (\g,f)$.
  \item If $a$ is a delay, then $a>\delta$, and $b\in\Sigma_!$. As $b$
    is controlled by the implementation, and there is no invariant in
    $\g$, $\sigma\cdot a\in \Outcome(s_0^\g,f)$. Hence $\sigma\cdot a\in
    \Behaviour(\g,f,\mI)$ and as $\sigma\cdot a\notin \Traces(\mS)$ and the
    determinization is exact, $\sigma\cdot a\notin \Traces_{\conf}(\ot)$ and
    $\g \after  \sigma\cdot a\in \VFail$. Hence $\mI \fails
    (\g,f)$.\qed
\end{itemize}
\end{proof}

\fi
\end{document}